\documentclass[fontsize=12pt]{article} 
\usepackage[utf8]{inputenc}
\usepackage{fancyhdr}
\usepackage{lastpage}
\usepackage{booktabs}
\usepackage{SLPformat}
\usepackage{multirow}
\usepackage{bbold}
\usepackage{array,multirow,makecell}
\usepackage{enumitem} 

\def\ulsib{\mbox{ULS-IB} }
\def\dulsib{2ULS-IB }
\def\dulsibr{\mbox{2ULS-IB$_R$} }
\def\dulsibs{\mbox{2ULS-IB$_S$} }
\def\dulsibsr{\mbox{2ULS-IB$_{SR}$} }
\def\dulsibsrNLS{\mbox{2ULS-IB$_{SR}$-NLS} }
\def\dulsibsNLS{\mbox{2ULS-IB$_{S}$-NLS} }
\def\dulsibrNLS{\mbox{2ULS-IB$_{R}$-NLS} }
\def\ulsNLS{\mbox{ULS-IB-NLS} }
\def\reg{regular }

\def\etal{{\it et al. }}
\def\ie{{\it i.e. }}

\title{Two-level lot-sizing with inventory bounds}
\author{Siao-Leu Phouratsamay, Safia Kedad-Sidhoum, Fanny Pascual \\[0.1cm]
{\normalsize Sorbonne Universit\'es, UPMC Univ Paris 06, CNRS, LIP6 UMR 7606, 
4 place Jussieu 75005 Paris.} \\
{\tt\normalsize \{siao-leu.phouratsamay,safia.kedad-sidhoum,fanny.pascual\}@lip6.fr}
}

\begin{document}
\maketitle
\thispagestyle{empty}

\begin{abstract}
We study a two-level uncapacitated lot-sizing problem with inventory bounds that occurs in a supply chain composed of a supplier and a retailer. The first level with the demands is the retailer level and the second one is the supplier level. The aim is to minimize the cost of the supply chain so as to satisfy the demands when the quantity of item that can be held in inventory at each period is limited. The inventory bounds can be imposed at the retailer level, at the supplier level or at both levels. We propose a polynomial dynamic programming algorithm to solve this problem when the inventory bounds are set on the retailer level. When the inventory bounds are set on the supplier level, we show that the problem is NP-hard. We give a pseudo-polynomial algorithm which solves this problem when there are inventory bounds on both levels. In the case where demand lot-splitting is not allowed, \ie each demand has to be satisfied by a single order, we prove that the uncapacitated lot-sizing problem with inventory bounds is strongly NP-hard. This implies that the two-level lot-sizing problems with inventory bounds are also strongly NP-hard when demand lot-splitting is considered.
\end{abstract}

\section{Introduction} 
\par We consider a two-level supply chain with a supplier and a retailer. 
The retailer has to satisfy a demand for a single item over a finite planning horizon of discrete periods. 
In order to satisfy the demand, the retailer has to determine a replenishment plan over the horizon, \ie  when and how many units to order. 
In order to satisfy the retailer's replenishment plan, the supplier has to determine a production plan. 
Ordering units induce a fixed ordering cost and a unit ordering cost for both actors. Carrying units in the inventory induce a unit holding cost for both actors as well.
Moreover, the quantity that can be held in inventory at each period can be limited, since  inventory bounds can be imposed at the retailer level, at the supplier level or at both levels. 
The cost of the supply chain is given by the sum of the supplier and the retailer total costs. 
The two-level Uncapacitated Lot-Sizing (2ULS) problem with inventory bounds consists in determining the order and the inventory quantities at each period for both replenishment and production plans in order to satisfy the external demand while minimizing the total cost of the supply chain.

\subsection*{Literature review}
\par For many practical applications, it is unreasonable to suppose that the inventory capacity is unlimited. 
In particular, the products that need temperature control or special storage facilities may have a limited storage capacity. This is for example the case in the pharmaceutical industry~\cite{AtamturkAlgo}. 
These constraints have led to the study of lot-sizing problems with inventory bounds.

\par The single level Uncapacitated Lot-Sizing problem with Inventory Bounds (ULS-IB) was first introduced by Love~\cite{Love1973}. He proves that the problem with piecewise concave ordering and holding costs
and backlogging can be solved using an $\mathcal{O}(T^3)$ dynamic programming algorithm. 
Atamt\"urk and K\"u\c{c}\"ukyavuz~\cite{AtamturkAlgo} study the ULS-IB problem under the cost structure assumed in Love's paper~\cite{Love1973}, considering in addition a fixed holding cost. 
They propose an $\mathcal{O}(T^2)$ algorithm to solve the problem. They also make a polyhedral study of the ULS-IB problem~\cite{AtamturkPoly} by considering two cost structures: linear holding costs, linear and fixed holding costs. They provide an exact separation algorithm for each problem. 
Toczylowski \cite{Toczylowski} addresses this problem by solving a shortest path problem in $\mathcal{O}(T^2)$ time. 
More recently, Hwang and van den Heuvel \cite{Hwang} propose an $\mathcal{O}(T^2)$ dynamic programming algorithm to solve the ULS-IB problem with backlogging and a concave cost structure by exploiting the so-called Monge property. 
Guti\'errez \etal \cite{Gutierrez2008} improved the time complexity by developing an algorithm that runs in $\mathcal{O}(T\mbox{log}T)$ using the geometric technique of Wagelmans and van Hoesel \cite{WagelmansULS}. 
However, van den Heuvel \etal \cite{vandenHeuvel2011} show that their algorithm does not provide an optimal solution for the ULS-IB problem. 
Liu~\cite{Liu2008} proposes an $\mathcal{O}(T^2)$ algorithm based on the geometric approach in~\cite{WagelmansULS} but \"Onal \etal \cite{Onal2012} prove that his algorithm does not compute an optimal  solution for the ULS-IB problem.

\par Zangwill~\cite{Zangwill} proposes an $\mathcal{O}(LT^4)$ dynamic programming algorithm for the multi-level uncapacitated lot-sizing problem (where $L$ is the number of levels). 
In particular, van Hoesel \etal observe that Zangwill's algorithm runs in $\mathcal{O}(T^3)$ when $L=2$~\cite{vanHoesel2005}.
More recently, Melo and Wolsey~\cite{Melo_Wolsey} improve this complexity by proposing an $\mathcal{O}(T^2 \mbox{log} T) $ dynamic programming algorithm. 
Zhang \etal \cite{Zhang2012} propose a polyhedral study of the multi-level lot-sizing problem where each level has its own external demand. 
They give an $\mathcal{O}(T^4)$ dynamic programming algorithm to solve the two-level problem. 
A few papers deal with the 2ULS problem with inventory bounds. Jaruphongsa \etal\cite{Jaruphongsa2004} study this problem with 
demand time window constraints and stationary inventory bounds at the supplier level. 
They impose some assumptions on the cost parameters (among them, the unit production cost is non-increasing). 
These assumptions make the problem solvable in $\mathcal{O}(T^3)$ using a dynamic programming algorithm. 
They also prove that when each demand is satisfied by a single dispatch, the problem is NP-hard. 
Hwang and Jung~\cite{HwangJung} propose a dynamic programming algorithm that solves the \dulsib problem with inventory bounds at the retailer level and concave costs in $\mathcal{O}(T^4)$ which has the same complexity than the one provided in this paper. 
However, contrary to their result, we present a dynamic programming algorithm based on some structural properties specific to the inventory bounds for which we give correctness proofs.

\subsection*{Contributions}
\par 
In this paper, we study the complexity of single-item 2ULS problems with inventory bounds.
We consider that either the supplier, the retailer, or both of them, have a limited inventory capacity.
A polynomial dynamic programming algorithm is provided to solve the problem with inventory bounds at the retailer level.
The problem is shown to be weakly NP-hard when the inventory bounds are imposed at the supplier level.
A complexity analysis for this class of problem is also proposed under the no lot-splitting assumption.
In the sequel, we will denote \dulsibr (resp. \dulsibs), the problem where at each period, the inventory quantity at the retailer (resp. supplier) 
level cannot exceed the inventory bound. Finally, the \dulsibsr problem is the problem where both the supplier and the retailer have a limited inventory capacity. 
\\

 
\par  This paper is organized as follows. A mathematical formulation for the single-item 2ULS problem with inventory bounds is provided in Section~\ref{sec:notations}. 
In Section~\ref{sec:2uls_ibr}, we propose a polynomial algorithm to solve the \dulsibr problem. In Section \ref{sec:2uls_ibs}, 
we prove that the \dulsibs problem is NP-hard. We also show that the \dulsibsr problem is solvable using a pseudo-polynomial algorithm in Section \ref{sec:2uls_ibsr}. 
Finally,  we prove in Section \ref{sec:2uls_nls} that these problems are strongly NP-hard under the no lot-splitting assumption 
where each demand has to be satisfied by a unique order.

\section{Mathematical formulations} \label{sec:notations} 
\par In this section, we describe the mathematical formulation of the 2ULS problem as well as the inventory bound constraints for the addressed problems. 

\par  Let $T$ be the number of periods over the planning horizon. 
We denote by $d_t$ the demand at each period $t$ for $t\in\{1,\cdots,T\}$. The retailer's (resp. supplier's) costs are defined by  
a fixed ordering cost $f^R_t$ (resp. $f^S_t$), a unit ordering cost $p^R_t$ (resp. $p^S_t$) and a unit holding cost $h^R_t$ (resp. $h^S_t$) for $t\in\{1,\cdots,T\}$. 
The retailer's (resp. supplier's) inventory bound at each period $t$ is denoted by $u^R_t$ (resp. $u^S_t$) for $t \in \{1,\cdots,T\}$.

\par  We denote by $x^R_t$ (resp. $x^S_t$) the quantity ordered by the retailer (resp. supplier) at period $t$, $s^R_t$ (resp. $s^S_t$) the retailer's (resp. supplier's) inventory level at the end of period 
$t$ and $y^R_t$ (resp. $y^S_t$) the retailer's (resp. supplier's) setup variable, which is equal to 1 if an order occurs at period $t$ at the retailer (resp. supplier) level and 0 otherwise. The 2ULS problem can be formulated as follows:
\begin{align}
\min & \sum \limits_{t=1}^{T} (f^S_ty^S_t + p^S_tx^S_t + h^S_ts^S_t + f^R_ty^R_t + p^R_tx^R_t + h^R_ts^R_t) \label{eq:obj}\\
\mbox{s.t. } & s^R_{t-1} + x^R_t = d_t + s^R_t \qquad \forall  t \in \{1,\ldots,T\}, \label{eq:flow2}\\
& s^S_{t-1} + x^S_t = x^R_t + s^S_t \qquad \forall  t \in \{1,\ldots,T\}, \label{eq:flow1}\\
& x^R_t \leq M^R_t y^R_t \qquad \forall  t \in \{1,\ldots,T\}, \label{lien_xy2}\\
& x^S_t \leq M^S_t y^S_t \qquad \forall  t \in \{1,\ldots,T\}, \label{lien_xy1} \\
& x^S, x^R, s^S, s^R \geq 0  \label{eq:duls_positive} \\
& y^S, y^R \in \{0,1\}^T  \label{eq:duls_binaire}
\end{align}
where $M^R_t = M^S_t = \sum_{i=t}^T d_i$. 
The supply chain total cost is given by (\ref{eq:obj}). 
Constraints~(\ref{eq:flow2}) (resp. (\ref{eq:flow1})) are the inventory balance constraints at the retailer (resp. supplier) level. 
The supplier demand is the amount ordered at the retailer level at each period $t$. 
Constraints~(\ref{lien_xy2}) and (\ref{lien_xy1}) force the setup variables to be equal to 1 if there is an order, \ie if $x^R_t > 0$ or $x^S_t > 0$ respectively. 

\par The 2ULS problem can be viewed as a fixed charge network flow problem (see Figure \ref{fig:network}) where the nodes represent the periods at each level. 
A source node is also considered in order to represent the total supplied quantity $\sum_{i=1}^T d_i$. 
For each node, the vertical inflows are the ordering quantities and the 
horizontal outflows represent the inventory quantities. In addition, arcs representing the external demand at each period at the retailer level are considered.  
In the  sequel, we will not represent the dummy node, and the arcs will be represented only if they are active (i.e. a vertical arc will be represented if the corresponding ordering quantity is positive, and a horizontal arc is represented if the corresponding inventory quantity is not null). \\

\begin{figure}[!ht]
    \centering
        \includegraphics[clip=true, scale = 0.9]{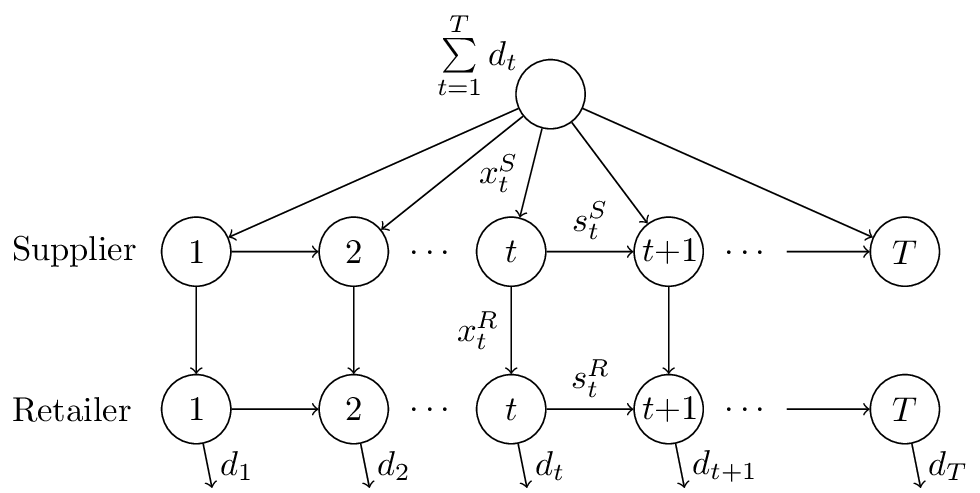}
        \caption[Fig]{The 2ULS problem as a fixed charge network flow.}
        \label{fig:network}
\end{figure} 

In addition to this classical problem, we introduce inventory bounds constraints. 
The inventory bounds constraints for the \dulsibr problem are given by: \begin{equation} s^R_t \leq u^R_t \qquad \forall t \in \{1,\ldots,T\}. \label{ibr_const} \end{equation}
The mathematical formulation can be strengthened by setting $M^R_t$ to $\min(d_t + u^R_t,\sum_{i=t}^T d_i)$ in the constraint~(\ref{lien_xy2}). 
Similarly, the inventory bounds constraints for the \dulsibs problem are given by: \begin{equation} s^S_t \leq u^S_t \qquad \forall t \in \{1,\ldots,T\}. \label{ibs_const}\end{equation}
Similarly, parameter $M^S_t$ can be replaced by $\min(d_t + u^S_t + u^R_t ,\sum_{i=t}^T d_i)$ in the constraint~(\ref{lien_xy1}). The mathematical formulation of the \dulsibsr problem is obtained by adding
the constraints (\ref{ibr_const}) and (\ref{ibs_const}) to the mathematical formulation of the 2ULS problem.  

\section{The \dulsibr problem } \label{sec:2uls_ibr}
\par  
To the best of our knowledge, the \dulsibr problem has not been studied in the literature yet. 
We present some structural properties of an optimal solution for the problem and propose an $\mathcal{O}(T^4)$ algorithm to solve it. 
Since the inventory bounds are only set at the retailer level, the superscript $R$ will be omitted in the inventory bound parameter $u^R_t$ that will be denoted by $u_t$. 

\par  Zangwill~\cite{Zangwill} shows that there exists an optimal solution for the 2ULS problem that verifies the Zero Inventory Ordering (ZIO) property at each level, \ie  $s^i_{t-1}x^i_t = 0$ for all $t \in \{1,\ldots,T\}$ and $i \in \{S,R\}$. 
As shown by~\cite{AtamturkPoly,Toczylowski}, the following assumption can be stated without loss of generality:
\begin{asump}\label{obs:ib_prop}
$u_{t-1} \leq u_t + d_t$ for all $t \in \{1,\ldots,T\}$.
\end{asump}

\subsection{Dominance properties}

\par  In this section, we propose some dominance properties in order to determine an efficient solving approach such that there exists an optimal solution for the \dulsibr problem that satisfies these properties. 

\par We know that the ZIO property does not hold for the ULS-IB problem \cite{Love1973,Liu2008}. 
Let us first show that for the \dulsibr problem, the cost of the best solution  in which  the ZIO property is fulfilled at the retailer level may be arbitrarily large compared to the cost of an optimal solution in which the ZIO property is no required.

\begin{prop}
For the \dulsibr problem, the cost of the best ZIO policy at the retailer level may be arbitrarily large compared to the cost of an optimal policy.
\end{prop}
\begin{proof}
\par Consider the following instance $\mathcal{I}$: $T=2$, $h^S = p^S = f^R = h^R = [0,0]$, $f^S = [0,1]$, $p^R = [0,1], d = [0,B+1]$ and $u = [B,B]$, where $B$ is a large constant. 
The best solution satisfying the ZIO property at the retailer level is given by $x^S = [d_2,0]$, $ x^R = [0,d_2] $. The corresponding cost is $B+1$ whereas 
the optimal non-ZIO solution is given by $x^S = [d_2,0]$, $x^R = [B,1]$ inducing a cost equals to $1$ (see Figure~\ref{fig:zio_ib}). \end{proof}
\begin{figure}[!ht]
    \centering
        \includegraphics[clip=true, scale = 0.9]{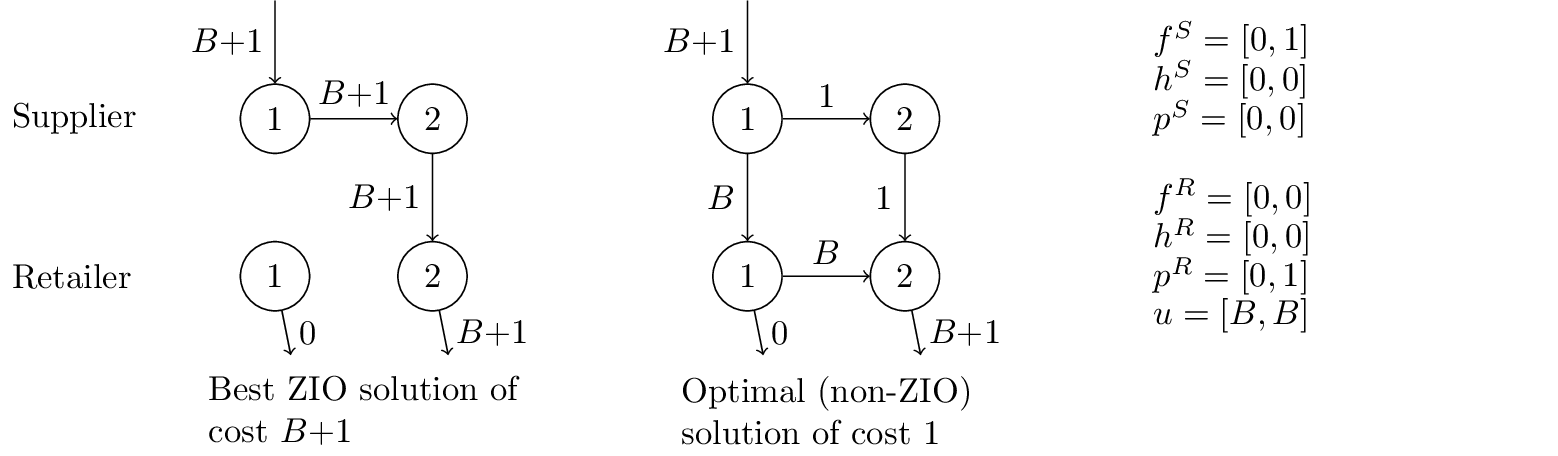}
        \caption[Fig]{Solutions for the instance $\mathcal{I}$ of the \dulsibr problem.}
        \label{fig:zio_ib}
\end{figure} 

\par  Let us now give the definition of a block (Definition~\ref{def:block}), previously introduced in \cite{AtamturkAlgo,AtamturkPoly} for the single level case. 

\begin{defn}[Subplan]
Let $i$ and $j$ be two periods such that $1 \leq i \leq j \leq T$. A subplan $[i,j]$ is a partial solution at the retailer level of the \dulsibr problem between the periods $i$ and $j$ defined by $x^R_i,\ldots,x^R_j$. 
\end{defn}

\begin{defn}[Regular subplan]
Let $i$ and $j$ be two periods such that $1 \leq i \leq j \leq T$. A \reg subplan $[i,j]$ is a subplan $[i,j]$ such that $s_{i-1}^R\in\{0,u_{i-1}\}$ and $s_j^R\in\{0,u_j\}$. 
\end{defn}

\begin{defn}[Block]\label{def:block}
Let $i$ and $j$ be two periods such that $1 \leq i \leq j \leq T$. Let $\alpha \in \{0,u_{i-1}\}$ and $\beta \in \{0,u_{j}\}$. A block $[i,j]^\alpha_\beta$ is a regular subplan $[i,j]$ where $s^R_{i-1} = \alpha$, $s^R_j = \beta$ and $0 < s^R_t < u_t$ for all $ t \in \{i,\ldots,j-1\}$. 
\end{defn}

In other words, a block $[i,j]^\alpha_\beta$ is a regular subplan $[i,j]$ where the inventory quantities for each period between $i$ and $j$ are strictly positive but not equal to the inventory bound. 
A \reg subplan is made of one or several blocks.

\begin{defn}[Order quantity]\label{def:block_prod}
Let $d_{tk} = \sum_{i=t}^k d_i$ be the cumulative demand between periods $t$ and $k$. The order quantity at the retailer level in a subplan $[i,j]$ is given by $X_{ij} = d_{ij} - s^R_{i-1} + s^R_j$.
\end{defn}
\par  Observe that for a block $[i,j]^\alpha_\beta$, $X_{ij} = d_{ij} - \alpha + \beta$. Thereafter, we give some properties observed by an optimal solution for the \dulsibr problem.

\begin{thm} \label{thm:one_prod_period}
Let $\mathcal{P}$ be the set of points that satisfy Constraints~(\ref{eq:flow2})-(\ref{ibr_const}) of the \dulsibr problem. 
A point $(y^R,x^R,s^R,y^S,x^S,s^S) \in \mathcal{P}$ is an extreme point if and only if: 
\begin{enumerate}
\item there is at most one ordering period in every block $[i,j]^\alpha_\beta$, for all $1 \leq i \leq j \leq T$, $\alpha \in \{0,u_{i-1}\}$, $\beta \in \{0,u_j\}$, 
\item the ZIO property holds at the supplier level. 
\end{enumerate}

\end{thm}

This theorem follows from the properties related to the optimal flows in a fixed-charged network with concave costs~\cite{Zangwill1968}. 
As the \dulsibr problem is a single source fixed-charged network with linear costs, Theorem~\ref{thm:one_prod_period} is a direct application of the characterization of extreme points in these networks.

\begin{prop} \label{prop:prod_period}
An extreme point of $\mathcal{P}$ satisfies the following properties at the retailer level for all $1 \leq i \leq j \leq T$, $\alpha \in \{0,u_{i-1}\}$, and $\beta \in \{0,u_j\}$:
\begin{enumerate}[label=(\roman*)]
\item If $[i,j]^0_\beta$ is a block and if there is an ordering period in this block, then this ordering period is $i$.  \label{prop:prod_period_i}
\item If $[i,j]^\alpha_{u_j}$ is a block and if there is an ordering period in this block, then this ordering period is $j$. \label{prop:prod_period_ii}
\end{enumerate}
\end{prop}

\begin{proof}
\noindent (i) 
If $d_i>0$, then period $i$ is necessarily an ordering period since $s^R_{i-1} = 0$. \\
If $d_i=0$ and there is an ordering period in the block $[i,j]^0_\beta$, then we have $s^R_i > 0$ and period $i$ is necessarily an ordering period since $s^R_{i-1} = 0$. \\

\noindent (ii) Assume that the (unique) ordering period is $k$ in the block $[i,j]^\alpha_{u_j}$ with $i\leq k \leq j-1$. From Assumption~\ref{obs:ib_prop}, we have $u_{k} \leq u_{k+1} + d_{k+1} \leq u_{k+2} + d_{k+2} + d_{k+1} \leq \dots \leq u_{j} + d_{j} + \sum_{i=k+1}^{j-1} d_i$. Thus $u_{k} \leq u_{j} + d_{k+1,j}$.
There are two possible cases:\\
Case 1: $u_{k} = u_j + d_{k+1,j}$. In this case, since $s^R_j = u_j$, then $s^R_{k} = u_{k}$ which is not possible since  $[i,j]^\alpha_{u_j}$ is a block.\\
Case 2: $u_{j-1} < u_j + d_j$. In this case, it is not possible to have $s^R_j = u_j$ without having an additional ordering period in the block, which contradicts Theorem~\ref{thm:one_prod_period}.
So, the ordering period has to be at period $j$ in a block $[i,j]^\alpha_{u_j}$.
\end{proof}

\par Using Theorem~\ref{thm:one_prod_period} and Property~\ref{prop:prod_period}, we propose a polynomial algorithm to solve the \dulsibr problem.

\subsection{Recursion formula} \label{sec:2uls_ibc}
\par In this section, we derive a polynomial backward dynamic programming algorithm to solve the \dulsibr problem. The rationale of this algorithm is to compute a block decomposition 
of the retailer's replenishment plan such that the total cost of the supply chain is minimized using the dominance properties of the optimal solutions of the problem.   

\par 
Let $i,j$ be two periods such that $1 \leq i \leq j \leq T$. Let us consider a \reg subplan $[i,j]$ of a solution of the \dulsibr problem. 
Notice that by definition $[i,j]$ is not necessarily a block unless property $0<s^R_k < u_k$ for all $k \in \{i,\ldots,j-1\}$ holds. 
Assume that at period $t$, an order quantity $X_{ij} = d_{ij} - s^R_{i-1} + s^R_j$ (see Definition~\ref{def:block_prod}) is available at the supplier level, \ie 
it is either ordered at period $t$ or stored at period $t-1$ assuming the ZIO policy. 
The aim is to decompose the \reg subplan $[i,j]$ into blocks satisfying Property~\ref{prop:prod_period}. 

\par An example is given in Figure~\ref{fig:subplan_ibr}. The graph represents subplans of a solution for an instance of the \dulsibr problem where $T=4$. 
At period $t=1$, a quantity $X_{11} = d_1+u_1$ is available and at period $t=2$, a quantity $X_{24} = d_{24} - u_1$ is available at the supplier level (it is also available at period $3$). 
In this example, $[2,4]$ is a \reg subplan composed of the blocks $[2,3]^{u_1}_0$ and $[4,4]^0_0$. 
\begin{figure}[H]
    \centering
        \includegraphics[clip=true, scale = 1]{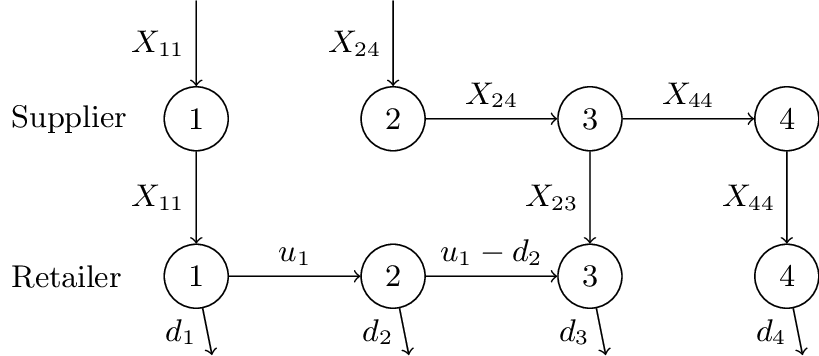}
        \caption[Fig]{Subplans decomposition for an instance of the \dulsibr problem where $T=4$.}
        \label{fig:subplan_ibr}
\end{figure}

\par The recursion formula will be defined in the following order given that $[i,j]$ is an interval of periods: 
the cost of a block ($\phi^{\alpha\beta}_{ijk}$), the cost functions required to compute the cost of a \reg subplan including the supplier's holding cost ($w^{\alpha\gamma\beta}_{tijk}$, $v^{\alpha\beta}_{tij}$ and $G^{\alpha\beta}_{tij}$) and finally the cost of the supply chain ($C^{\alpha\beta}_{tij}$).

\subsubsection{Computation of the cost of a block}

\par Let $\phi^{\alpha\beta}_{ijk}$ be the cost of satisfying the demands of a block $[i,j]^\alpha_\beta$ with a single ordering at period $k$ if it exists (Theorem~\ref{thm:one_prod_period}). 
We will denote by $\phi^{\alpha\beta}_{ij-}$ the cost of the block $[i,j]^\alpha_\beta$ without an ordering period. 

\par 
Using Property~\ref{prop:prod_period}\ref{prop:prod_period_i}, the cost $\phi^{0\beta}_{ijk}$ is defined as follows:
$$
\phi^{0\beta}_{ijk} = 
\left\{
\begin{array}{rl}
f^R_i + p^R_i X_{ij} + \sum \limits_{n=i}^{j} h^R_n (d_{n+1,j}+\beta), & \textrm{ if } k=i \textrm{ and } 0 < d_{ij} + \beta \leq u_i\\
0, & \textrm{ if } i=j \textrm{ and } d_{ij}+\beta = 0 \\
+\infty, & \textrm{ otherwise.}
\end{array}
\right.
$$

\par If $[i,j]^{u_{i-1}}_{u_j}$ is a block such that $u_{i-1} = d_{ij} + u_j$, then the demands of this block can be satisfied without setting any order between period $i$ and $j$ since $s^R_{i-1} = u_{i-1}$. 
Otherwise, using Property~\ref{prop:prod_period}\ref{prop:prod_period_ii}, a quantity $X_{ij}$ has to be ordered at period $j$. 
Moreover, we have to ensure that the inventory bounds constraints are not violated, 
and that the demands $d_{i,j-1}$ can be covered by the inventory quantity at the end of period $i-1$, \ie $u_{i-1} > d_{i,j-1}$.
Thus, the cost $\phi^{u_{i-1}u_j}_{ijk}$ of the block is given by:
$$
\phi^{u_{i-1}u_j}_{ijk} = 
\left\{
\begin{array}{rl}
\sum_{n=i}^{j} h^R_n (u_{i-1}-d_{in}+u_j),&\textrm{if } u_{i-1} = d_{ij} + u_j \\
f^R_k + p^R_k X_{ij} + \sum_{n=i}^{j-1} h^R_n(u_{i-1}- d_{in})+ h^R_j u_j, &  \\
\textrm{if } k=j \mbox{ and}&d_{ij} + u_j > u_{i-1} > d_{i,j-1} \\
+\infty,&\textrm{otherwise.}
\end{array}
\right.
$$

\par In a block $[i,j]^{u_{i-1}}_0$, if $u_{i-1} < d_{ij}$, then the quantity $X_{ij}$ can be ordered at any period $k$ between $i$ and $j$. 
In this case, we have to ensure that the inventory bounds constraints are not violated, that the inventory quantity $u_{i-1}$ covers the demands before period $k$ ($u_{i-1} > d_{i,k-1}$), 
and that the demands after period $k$ can be satisfied ($u_k \geq d_{k+1,j}$). The cost $\phi^{u_{i-1}0}_{ijk}$ is then given by:
$$
\phi^{u_{i-1}0}_{ijk} = 
\left\{ 
\begin{array}{rl}
f^R_k + p^R_k X_{ij} + \sum_{n=i}^{k-1}h^R_n(u_{i-1} - d_{in}) + \sum_{n=k}^{j-1} h^R_n d_{n+1,j}, & \\
 \textrm{ if } d_{ij} > u_{i-1} > d_{i,k-1} \textrm{ and}&u_k\geq d_{k+1,j} \\
\sum_{n=i}^{j-1}h^R_n(u_{i-1}-d_n),&\textrm{if } u_{i-1}=d_{ij} \\
+\infty, & \textrm{otherwise.}
\end{array}
\right.
$$

\subsubsection{Computation of the cost of a regular subplan}

\par Let $G^{\alpha\beta}_{tij}$ be the optimal cost to cover the demands $d_{ij}$ of the regular subplan $[i,j]$ where a quantity $X_{ij}$ is available at period $t$ at the supplier level with $1 \leq i \leq j \leq T$ and $1 \leq t \leq j$, $s^R_{i-1} = \alpha$ and $s^R_j = \beta$. 
Computing $G^{\alpha\beta}_{tij}$ requires the computation of the costs of the blocks that compose the subplan $[i,j]$. 
Therefore, in order to compute $G^{\alpha\beta}_{tij}$ efficiently, we first need to find a smart decomposition of the subplan $[i,j]$.

\par Let $w^{\alpha\gamma\beta}_{tijk}$ be the optimal cost of the subplan $[i,j]$ where a quantity $X_{ij}$ is available at period $t$ at the supplier level, $s^R_{i-1} = \alpha$, $s^R_j = \beta$ and $k$ is the ordering period of the first block $[i,l]^\alpha_\gamma$ of $[i,j]$, $i\leq l <  j$, $\gamma \in \{0,u_l\}$ (see Figure~\ref{fig:ibr_W}). 
The aim is to find a period $l$ ($k \leq l < j$) such that $[i,l]^\alpha_\gamma$ is the first block of the \reg subplan $[i,j]$ with an order at period $k$ in an optimal solution. 
The cost $w^{\alpha\gamma}_{tijk}$ is given by:
\begin{equation}
w^{\alpha\gamma\beta}_{tijk} = \min \limits_{k\leq l <j} \{ \phi^{\alpha\gamma}_{ilk} + G^{\gamma\beta}_{k,l+1,j} + \sum_{p=t}^{k-1} h^S_{p} X_{ij} \}.
\label{eq:w}
\end{equation}
The first term $\phi^{\alpha\gamma}_{ilk}$ in Equation~(\ref{eq:w}) represents the cost of satisfying the demands of the block $[i,l]^\alpha_\gamma$ with an ordering at period $k$. 
The second term $G^{\gamma\beta}_{k,l+1,j}$ in Equation~(\ref{eq:w}) represents the optimal cost for satisfying the demands of the \reg subplan $[l+1,j]$ assuming that the quantity $X_{l+1,j}$ is available at period $k$ at the supplier level. 
Finally, the last term $\sum_{p=t}^{k-1} h^S_{p} X_{ij}$ represents the cost of carrying $X_{ij}$ units from period $t$ to period $k$ at the supplier level.

\par 
A representation of the cost $w^{\alpha\gamma\beta}_{tijk}$ is depicted in Figure~\ref{fig:ibr_W}. There are $X_{ij}$ units available at period $t$ at the supplier level. 
At the retailer level, $[i,l]^\alpha_\gamma$ is the first block of the \reg subplan $[i,j]$, where $s^R_{i-1} = \alpha$ and $s^R_j = \beta$. 
A quantity $X_{il}$ is ordered at period $k$ in this block. 
At the supplier level, a quantity $X_{ij}$ is stored from period $t$ to period $k$ and an amount $X_{l+1,j}$ is available at period $k$ to satisfy the demands of the subplan $[l+1,j]$. 
The different terms of $w^{\alpha\gamma\beta}_{tijk}$ are shown in Figure~\ref{fig:ibr_W}. 
\begin{figure}[H]
    \centering
        \includegraphics[clip=true, scale = 1]{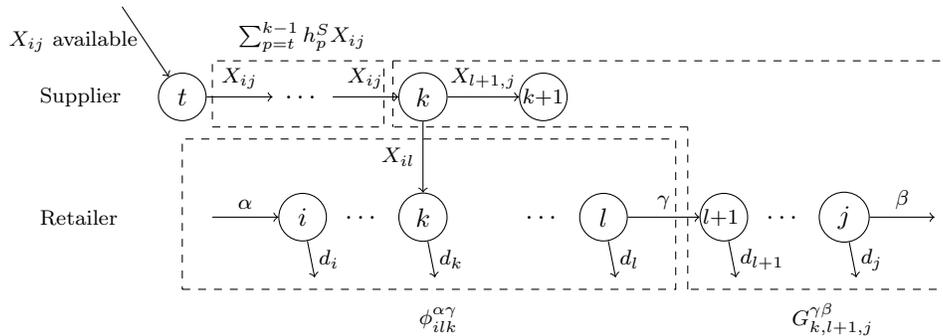}
        \caption[Fig]{Illustration of the cost $w^{\alpha\gamma\beta}_{tijk}$ where $[i,j]$ is a \reg subplan and an amount $X_{ij}$ is available at the supplier level at period $t$.}
        \label{fig:ibr_W}
\end{figure}

\par Let $v^{\alpha\beta}_{tij}$ be the minimum cost of a \reg subplan $[i,j]$ composed of at least two blocks such that $s^R_{i-1} = \alpha$ and $s^R_j = \beta$ and assuming that a quantity $X_{ij}$ is available at period $t$ at the supplier level, $1 \leq t \leq j$. 
In the sequel, we denote by $l^* = \mbox{argmin } w^{\alpha\gamma\beta}_{tijk}$. 
The first block of the \reg subplan $[i,j]$ is $[i,l^*]^\alpha_\gamma$ and the ordering period in this block is $k$ if it exists. 
The cost $v^{\alpha\beta}_{tij}$ is given by:
\begin{equation}
v^{\alpha\beta}_{tij} = \min \limits_{i\leq k <j; \gamma \in \{0,u_{l^*}\}} \{ w^{\alpha\gamma\beta}_{tijk} \}.
\end{equation} 

\par From the definition of the cost $v^{\alpha\beta}_{tij}$, we can then compute the cost $G^{\alpha\beta}_{tij}$ where the fixed ordering cost $f^S$ at the supplier level is not included.
The cost $G^{\alpha\beta}_{tij}$ is given by: 
\begin{equation}
G^{\alpha\beta}_{tij} = 
\left\{
\begin{array}{rl}
\min \text{{\Large \{}} \min \limits_{t \leq k \leq j} \{ \phi^{\alpha\beta}_{ijk} + \sum_{l=t}^{k-1} h^S_l X_{ij} \}, v^{\alpha\beta}_{tij} \text{{\Large \}}}, & \textrm{ if } X_{ij} > 0 \\
\phi^{\alpha\beta}_{ij-}, & \textrm{ if } X_{ij} = 0.
\end{array}
\right.
\label{eq:G}
\end{equation}
 In Equation~(\ref{eq:G}), the term $\min_{t \leq k \leq j} \{ \phi^{\alpha\beta}_{ijk} + \sum_{l=t}^{k-1} h^S_l X_{ij} \}$ represents the optimal cost of the \reg subplan $[i,j]$ 
 when it is made of a single block, and $v^{\alpha\beta}_{tij}$ is the optimal cost of $[i,j]$ when it is composed of at least two blocks.  
If $[i,j]^\alpha_\beta$ is not a block or if $k < i$ then the cost $\phi^{\alpha\beta}_{ijk}$ will be equal to $+\infty$. Moreover, $G^{\alpha\beta}_{tij}$ will be equal to $+\infty$ if $i > j$ or $t > j$.

\subsubsection{Computation of the cost of the supply chain}

\par Let $C^{\alpha\beta}_{tij}$ be the optimal cost of the supply chain for satisfying the demands $d_{ij}$ of the \reg subplan $[i,j]$ where the first ordering period at the supplier level is larger than or equal to $t$, with $1\leq i\leq j \leq T$, $1\leq t\leq j$, $\alpha \in \{0,u_{i-1}\}$ and $\beta \in \{0,u_j\}$. The total ordering quantity of the subplan is equal to $X_{ij}$. 
The aim is to determine the ordering periods satisfying the ZIO property at the supplier level in order to satisfy the demands of the \reg subplan $[i,j]$. 

\par If $X_{ij} = 0$, then no order is required at the supplier level. The cost is then equal to the cost $G^{\alpha\beta}_{tij}$ of the subplan $[i,j]$ with $s^R_{i-1} = \alpha$ and $s^R_j = \beta$:
\begin{equation} C^{\alpha\beta}_{tij} = G^{\alpha\beta}_{tij}. \end{equation}

\par If $X_{ij} > 0$, then the quantity $X_{ij}$ is completely or partially ordered at period $t$ or at a subsequent period if no order occurs at period $t$ at the supplier level. 
The cost $C^{\alpha\beta}_{tij}$ is given by the following equation where $\mathbb{1}(x)$ is a function equals to 0 if $x=0$ and $+\infty$ otherwise (see Figure~\ref{fig:ibr_C}).
\begin{align}
C^{\alpha\beta}_{tij} = \min \text{{\Large \{}} & C^{\alpha\beta}_{t+1,i,j}, 
f^S_t + p^S_t X_{ij} + G^{\alpha\beta}_{tij}, \label{eq:Cs}\\
& \min \limits_{i\leq l <j; \gamma \in \{0,u_l\}} \{ \min(f^S_t + p^S_t X_{il}, \mathbb{1}(X_{il})) + G^{\alpha\gamma}_{til} + C^{\gamma\beta}_{t^*+1,l+1,j} \} \text{{\Large \}}}, \nonumber
\end{align}
where $t^*$ is the last ordering period at the retailer level in the regular subplan $[i,l]$ ($t^*$ is determined and stored when the cost $G^{\alpha\gamma}_{til}$ is computed). 
The period $t^*+1$ is the earliest ordering period from which the supplier can order for satisfying the demands of the regular subplan $[l+1,j]$. If there is no ordering period in the \reg subplan $[i,l]$, then we set $t^*=t$.\\
The first term in Equation~(\ref{eq:Cs}) corresponds to the case where there is no order at period $t$ at the supplier level. The second term in Equation~(\ref{eq:Cs}) corresponds to the case where a quantity $X_{ij}$ is ordered at period $t$ at the supplier level. Finally, the last term in Equation~(\ref{eq:Cs}) represents the case where the quantity $X_{ij}$ is partially ordered at period $t$ at the supplier level: a quantity $X_{il} > 0$ is ordered at period $t$ to satisfy the demands of the \reg subplan $[i,l]$ with $i \leq l < j$. Because of the ZIO property at the supplier level, the supplier orders the quantity $X_{l+1,j}$ after period $t^*$. 

\par 
A representation of the last term of the cost $C^{\alpha\beta}_{tij}$ is provided in Figure~\ref{fig:ibr_C}.
In this figure, a quantity $X_{il}$ of units is ordered at period $t$ at the supplier level for satisfying the demands of the \reg subplan $[i,l]$ where $s^R_{i-1} = \alpha$ and $s^R_l = \gamma \in \{0,u_l\}$. 
The period $t^*$ corresponds to the last ordering period in the \reg subplan $[i,l]$. 
Since the ZIO property holds at the supplier level, we know that $s^S_{t^*}=0$.
Then, the next likely candidate for an ordering period at the supplier is the 
period $t^*+1$ if it exists. 
The components in the definition of the cost $C^{\alpha\beta}_{tij}$ are depicted in the figure. 
\begin{figure}[H]
    \centering
        \includegraphics[clip=true, scale = 1]{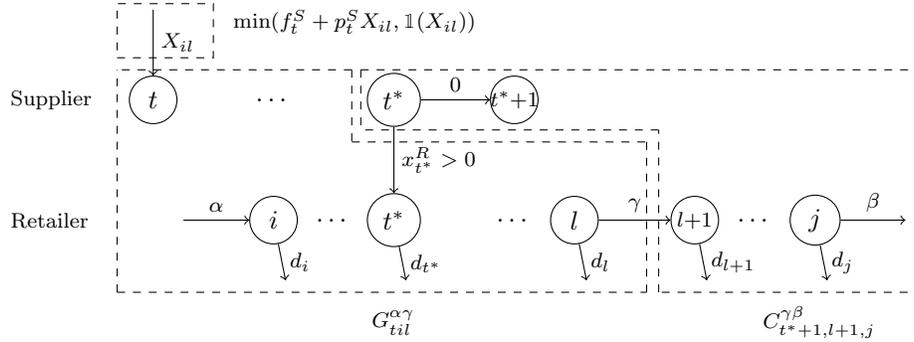}
        \caption[Fig]{Illustration of the cost $C^{\alpha\beta}_{tij}$  where $[i,j]$ is a \reg subplan and $t$ is an ordering period at the supplier level. }
        \label{fig:ibr_C}
\end{figure}

\noindent \textbf{Optimal cost.} 
\par The optimal cost of satisfying the demands of the \reg subplan $[1,T]$ is given by $C^{00}_{11T}$ since  $s^R_0=0$, $s^R_T=0$ and the earliest order period at the supplier level is $t=1$.

\subsubsection{Complexity analysis}
\par  A pre-processing phase will consist in the computation of $d_{1j}$ for all $j \in \{1,\ldots,T\}$ in $\mathcal{O}(T)$. 
Therefore, each $d_{ij}$ for all $i,j \in \{1,\ldots,T\}$ can be computed in constant time. 
Moreover, the holding costs required in the computation of each cost component is pre-computed and stored in $\mathcal{O}(T^2)$. \\
Therefore, the cost $\phi^{\alpha\beta}_{ijk}$ can be computed and stored in $\mathcal{O}(T^3)$ for all $i,j,k \in \{1,\ldots,T\}$. 
Besides, it takes $\mathcal{O}(T^4)$ time to compute and to store the costs $G^{\alpha\beta}_{tij}$ and $v^{\alpha\beta}_{tij}$. 
Finally, the cost $w^{\alpha\gamma\beta}_{tijk}$ is computed in $\mathcal{O}(T^5)$, and then the time complexity of the dynamic programming algorithm 
based on the recursion formula~(\ref{eq:Cs}) to compute $C^{00}_{11T}$ is $\mathcal{O}(T^5)$. 

\par In what follows, we show how the time complexity of computing the cost $w^{\alpha\gamma\beta}_{tijk}$ can be improved from $\mathcal{O}(T^5)$ to $\mathcal{O}(T^4)$ by generalizing the result of Atamt\"urk and K\"u\c{c}\"ukyavuz~\cite{AtamturkAlgo} for the 2ULS case.
To this end, we first need to recall the observation of Atamt\"urk and K\"u\c{c}\"ukyavuz~\cite{AtamturkAlgo} for the retailer level.
We provide a detailed explanation of the observation in~\ref{appendix:atamturk}.

\begin{obs} \label{obs1}
For all $1 \leq i \leq k \leq j \leq T$, $\alpha \in \{0,u_{i-1}\}$ and $\beta \in \{0,u_j\}$, we have:
\begin{enumerate}[label=(\roman*)]
\item if $\phi^{\alpha\beta}_{ijk} = +\infty$, then $\phi^{\alpha'\beta}_{i-1,j,k} = +\infty$ where $\alpha' = 0$ if $\alpha = 0$ and $\alpha' =u_{i-2}$ if $\alpha = u_{i-1}$.
\item if $\phi^{0\beta}_{ijk} \neq +\infty$, then $\phi^{0\beta}_{i-1,j,k} = +\infty$.
\item if $\phi^{u_{i-1} \beta}_{ijk} \neq +\infty$, then:\\
$\phi^{u_{i-2} \beta}_{i-1,j,k} = \left\{
\begin{aligned} 
& \phi^{u_{i-1} \beta}_{ijk} + \Delta_{1}, \\ 
&\quad \textrm{ if } u_{i-2} > d_{i-1,k-1}, u_{i-2} < u_{i-1}+d_{i-1} \textrm{ and } d_{ij} + \beta > u_{i-1}\\  
& \phi^{u_{i-1} \beta}_{ijk} + \Delta_{2}, \\
&\quad \textrm{ if } u_{i-2} > d_{i-1,k-1}, u_{i-2} < u_{i-1}+d_{i-1} \textrm{ and } d_{ij} + \beta = u_{i-1} \\
&  +\infty,  \mbox{ otherwise}
\end{aligned}
\right.
$
\\
\noindent where $\Delta_{1}=h^R_{i-2}u_{i-2} + (p^R_k - \sum_{l=i-1}^{k-1}h^R_l)(u_{i-1} - u_{i-2} + d_{i-1})$ and $\Delta_{2} =  f^R_k + \Delta_1$. 

\end{enumerate}
\end{obs}

\par 
The ordering quantity $X_{i-1,j}$ of a block $[i-1,j]^{u_{i-2}}_\beta$ consists of the ordering quantity  $X_{ij}$ in a block $[i,j]^{u_{i-1}}_\beta$ and the quantity $u_{i-1} - u_{i-2} + d_{i-1}$.
This implies that by replacing $X_{i-1,l}$ in the definition of the cost $w^{\alpha\gamma\beta}_{t,i-1,j,k}$ by $X_{il} + u_{i-1} - u_{i-2} + d_{i-1}$, the cost $w^{\alpha\gamma\beta}_{t,i-1,j,k}$ can be computed from $w^{\alpha\gamma\beta}_{tijk}$ independently of period $l$ by using the observation above. Therefore, 
for all $1 \leq i \leq T$ and given $k,t,j$, with $i \leq k \leq j\leq T$ and $1\leq t\leq j$, for $\alpha \in \{0,u_{i-1}\}$, $\gamma \in \{0,u_{l^*}\}$ and $\beta \in \{0,u_j\}$, the cost $w^{\alpha\gamma\beta}_{tijk}$ can be done in $\mathcal{O}(T)$ using the following equations:  
\begin{enumerate}[label=(\roman*)]
\item $w^{0\gamma\beta}_{t,i-1,j,k} = +\infty$ for any value of $w^{0\gamma\beta}_{tijk}$ 
\item $w^{u_{i-2}\gamma\beta}_{t,i-1,j,k}= 
\left\{
\begin{aligned}
& w^{u_{i-1}\gamma\beta}_{tijk} + \Delta_1 + \sum_{l=t}^{k-1}h^S_l (u_{i-1} - u_{i-2} + d_{i-1}), \\
& \qquad \textrm{if } u_{i-2} > d_{i-1,k-1} \textrm{ and } d_{il^*} + \beta > u_{i-1} \\
& w^{u_{i-1}\gamma\beta}_{tijk} + \Delta_2 + \sum_{l=t}^{k-1}h^S_l (u_{i-1} - u_{i-2} + d_{i-1}), \\
& \qquad \textrm{if } u_{i-2} > d_{i-1,k-1} \textrm{ and } d_{il^*} + \beta = u_{i-1} \\
& +\infty,  \hspace{1.5cm}\textrm{ otherwise}
\end{aligned}
\right.
$
\end{enumerate}
where $l^*= \mbox{argmin } w^{u_{i-1}\gamma\beta}_{tijk}$. \\

\par Consequently, for all periods $i,k,t,j$ such that $1\leq i \leq k \leq j \leq T$ and $1\leq t \leq j$, the cost $w^{\alpha\gamma\beta}_{tijk}$ can be computed in $\mathcal{O}(T^4)$. This implies that the algorithm which solves the \dulsibr problem runs in $\mathcal{O}(T^4)$.

\section{The \dulsibs problem} \label{sec:2uls_ibs}

\par Jaruphongsa \etal \cite{Jaruphongsa2004} propose a polynomial time algorithm to solve the \dulsibs problem with demand time window constraints and stationary inventory bounds. 
They consider that $h^S \leq h^R$ and that the fixed ordering cost and the unit ordering cost are decreasing. 
These specific costs make the problem solvable in polynomial time. 
In this section, we consider the \dulsibs problem under a general cost structure and we prove that this problem is NP-hard.

\begin{thm} \label{thm:dulsibs_np_diff}
The \dulsibs problem is NP-hard.
\end{thm}

\begin{proof} 
\par We prove that the \dulsibs problem is NP-hard through a reduction from the subset sum problem, which is an NP-complete problem~\cite{gareyjohnson}. An instance of the subset sum problem is given by an integer $S$ and a set $\mathcal{S}$ of $n$ integers $(a_1,\ldots,a_n)$. The question is: does there exist a subset $\mathcal{A} \subseteq \mathcal{S}$ such that $\sum_{a_i \in \mathcal{A}} a_i = S$?

\par We transform an instance of the subset sum problem into an instance of the \dulsibs problem in the following way:
\begin{itemize}
\item[-] $T = 2n + 1$. Let us denote by $\mathcal{T}_1$ (resp. $\mathcal{T}_2$) the set of odd (resp. even) periods in the set $\{1,\ldots,2n\}$. 
\item[-] $d_t = 0$ for all $t \in \mathcal{T}_1 \cup \mathcal{T}_2$, $d_T = S$
\item[-] $f^S_t = 1$ for all $t \in \mathcal{T}_1$, $f^S_t = 2S$ for all $t \in \mathcal{T}_2 \cup \{T\}$ \\
$f^R_t = 2S$ for all $t \in \mathcal{T}_1 \cup \{T\}$, $f^R_t = 0$ for all $t \in \mathcal{T}_2$
\item[-] $h^S_t = h^R_t = 0$ for all $t \in \{1,\ldots,T\}$
\item[-] $p^R_t = 0$ for all $t \in \{1,\ldots,T\}$ \\
$p^S_t = 1-1/a_{\left\lceil \frac{t}{2} \right\rceil} $ for all $t \in \mathcal{T}_1 \cup \mathcal{T}_2$, $p^S_T = 0$
\item[-] $u^S_t = a_{\left\lceil \frac{t}{2} \right\rceil}$ for all $t \in \mathcal{T}_1 \cup \mathcal{T}_2$, $u^S_T = 0$
\end{itemize}
A representation of this instance is given in Figure \ref{fig:dulsibs_proof_inst}. The fixed ordering costs and the unit ordering costs of the supplier (resp. retailer) are indicated at the top (resp. bottom). At the supplier level, the quantities on the horizontal edges represent the inventory bounds.
\begin{figure}[H]
    \centering
        \includegraphics[clip=true, scale = 0.95]{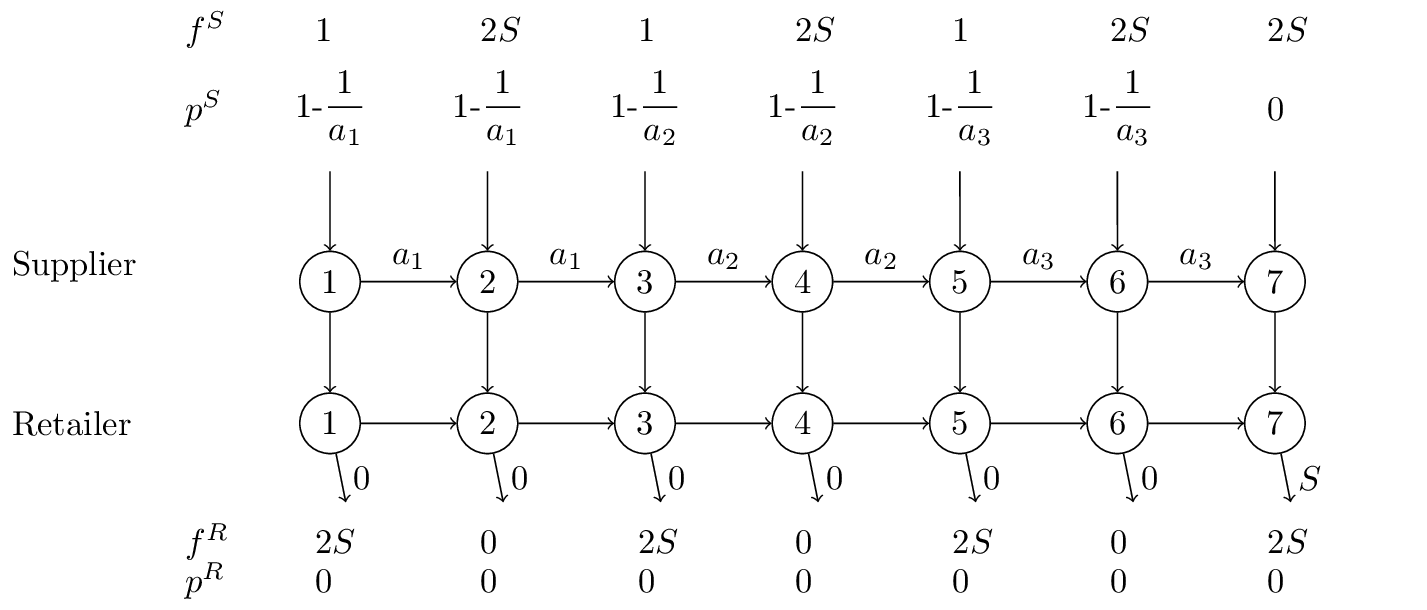}
        \caption{Instance $\mathcal{A}$ of the 2ULS-IB$_S$ problem in the proof of Theorem \ref{thm:dulsibs_np_diff} with $n=3$, $\mathcal{S} = \{a_1,a_2,a_3\}$.}
        \label{fig:dulsibs_proof_inst}
\end{figure}

\begin{obs} \label{obs:dulsibs_proof}
\normalfont  
Note that if we order $x^S_t = a_{\left\lceil \frac{t}{2} \right\rceil}$ at period $t \in \mathcal{T}_1$ then the total ordering cost is equal to $f^S_t + p^S_tx^S_t = 1 + (1-1/a_{\left\lceil \frac{t}{2} \right\rceil})a_{\left\lceil \frac{t}{2} \right\rceil}$ which is exactly equal to $x^S_t$ (in this case, the average  cost of ordering one unit is equal to 1). If $x^S_t < a_{\left\lceil \frac{t}{2} \right\rceil}$ at period $t \in \mathcal{T}_1$, then we have that the total ordering cost $f^S_t + p^S_t x^S_t = 1+ (1-1/a_{\left\lceil \frac{t}{2} \right\rceil})x^S_t =  x^S_t + 1-x^S_t/a_{\left\lceil \frac{t}{2} \right\rceil} > x^S_t$ (in this case, the average cost of ordering one unit is larger than 1). From this observation, let us prove that there exists a solution for the \dulsibs problem of cost at most $S$ if and only if there exists a solution for the subset sum problem.
\end{obs}

\par Assume that there exists a solution $\mathcal{A}$ of the subset sum problem. The following solution for the \dulsibs problem is of cost at most $S$: for each element $a_i$ in the set $\mathcal{A}$, the supplier orders a quantity $a_i$ at period $t=2i-1$ and store it until period $t+1$ (see Figure~\ref{fig:dulsibs_proof_sol}). The inventory bound is not exceeded since it is exactly equal to $a_i$. From Observation~\ref{obs:dulsibs_proof} above, the cost of ordering $a_i$ units for each $a_i \in \mathcal{I} $ at the supplier level is equal to $a_i$. Since $\sum_{a_i \in I} a_i = S$, the total cost at the supplier level is $S$. At period $t=2i$, 
the retailer orders all the units and store them until period $T$. Since $f^R_t = 0$ for all $t \in \mathcal{T}_2$ and $h^R_t = p^R_t = 0$ for all $t \in \{1,\ldots,T\}$, the total cost at the retailer level is equal to 0. So, there exists a solution for the \dulsibs problem of cost $S$.
\begin{figure}[H]
    \centering
        \includegraphics[clip=true, scale = 0.95]{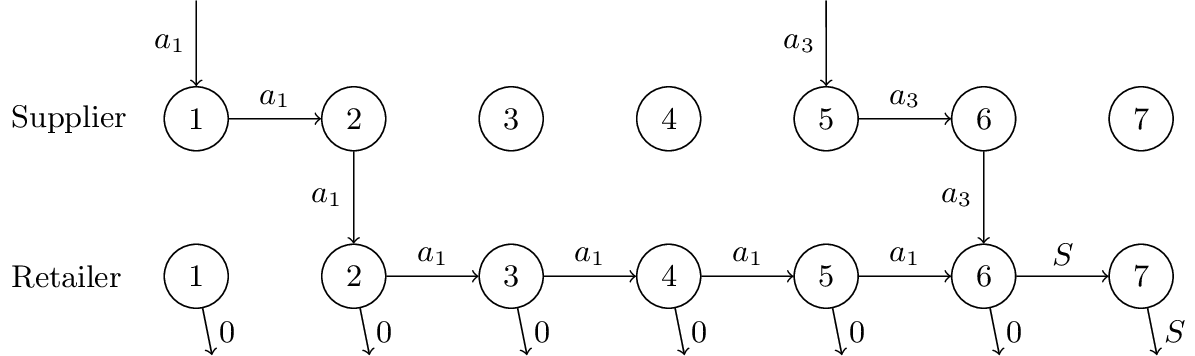}
        \caption{Solution for the 2ULS-IB$_S$ problem in the proof of Theorem \ref{thm:dulsibs_np_diff} with $n=3$, $\mathcal{S} = \{a_1,a_2,a_3\}$ and $a_1+a_3 = S$.}
        \label{fig:dulsibs_proof_sol}
\end{figure}

\par Assume that there exists a solution for the \dulsibs problem with a cost of at most $S$. Since $f^S_t = 2S$ for all $t \in \mathcal{T}_2$, the supplier has to order at period $t \in \mathcal{T}_1$, otherwise the cost will exceed $S$. Likewise, since $f^R_t = 2S$ for all $t \in \mathcal{T}_1$, the retailer has to order at period $t \in \mathcal{T}_2$. In order to not exceed the inventory bounds, the supplier can store at most $u^S_t = a_{\left\lceil \frac{t}{2} \right\rceil}$ units from period $t$ to period $t+1$. Thus, the quantity ordered by the supplier at period $t\in \mathcal{T}_1$ is at most $a_{\left\lceil \frac{t}{2} \right\rceil}$. At period $t \in \mathcal{T}_2$, the retailer orders the units in the supplier's inventory and stores them until period $T$ 
with a cost equal to 0. From Observation~\ref{obs:dulsibs_proof}, if the supplier orders at period $t$, then $x^S_t = a_{\left\lceil \frac{t}{2} \right\rceil}$ 
(this is the only way to order one unit with a cost of at most 1 so that the total cost is at most $S$). Thus, $S = \sum_{t \in \mathcal{T}} a_{\left\lceil \frac{t}{2} \right\rceil}$ where $\mathcal{T}$ is the set of periods where the supplier orders. This implies that there exists a solution to the subset sum problem.
\end{proof}

\par
The related lot-sizing problem with production capacity constraints instead of inventory bounds has been proved to be NP-hard~\cite{Florian1980}. 
It is worth noticing that in the proof of Theorem~\ref{thm:dulsibs_np_diff}, the inventory bound at period $t$ acts as a production capacity since the supplier cannot 
supply at $t$ the ordered units.

\section{The \dulsibsr problem}\label{sec:2uls_ibsr}
\par We have proved that the \dulsibs problem is NP-hard. By setting $u^R_t = \sum_{t=1}^T d_t$, we can transform an instance of the \dulsibs problem into an instance of the \dulsibsr problem. 
Thus, the \dulsibsr problem is at least as hard as the \dulsibs problem. 
In this section, we describe a pseudo-polynomial dynamic programming algorithm to solve the \dulsibsr problem.
This proves that this problem is not strongly NP-hard.

\par  Let $s^R_{t} \in \{0,1,\ldots,u^R_{t}\}$ (resp. $s^S_{t}\in \{0,1,\ldots,u^S_{t}\}$) be the inventory quantity available at the end of period $t$ at the retailer (resp. supplier) level.
The principle of the algorithm is to consider all the possible values of the inventory quantity $s^R_t$ (resp. $s^S_t$) at the retailer (resp. supplier) level. 
Notice that the ZIO property does not hold neither at the supplier nor at the retailer levels for the \dulsibsr problem. 

\par Let $C^i_t(X)$ be the cost of ordering $X$ units at level $i \in \{R,S\}$ at period $t$, where the level $R$ (resp. $S$) corresponds to the retailer (resp. supplier) level.
The cost $C^i_t(X)$ is given by:
$$
C^i_t(X) = \left\{
\begin{aligned}
f^i_t + p^i_t X, & \mbox{ if } X > 0 \\
0, & \mbox{ otherwise.}
\end{aligned}
\right.
$$ 

\par We define $V_t(s^R_{t-1},\bar{s},s^S_{t-1},s)$ as the cost of satisfying the demand $d_t$ when:
\begin{itemize}
\item[-] $s^R_{t-1}$ (resp. $s^S_{t-1}$) units are stored at period $t-1$ and $\bar{s}$ (resp. $s$) units are stored at period $t$ at the retailer (resp. supplier) level,
\item[-] $X^R=\bar{s}+d_t-s^R_{t-1}$ (resp. $X^S = s + X^R - s^S_{t-1})$ units are ordered at period $t$ at the retailer (resp. supplier) level. 
\end{itemize}  
The cost $V_t(s^R_{t-1},\bar{s},s^S_{t-1},s)$ is defined by:
$$
V_t(s^R_{t-1},\bar{s},s^S_{t-1},s) = \left\{
\begin{aligned}
C^R(X^R) + C^S(X^S), & \mbox{ if } \bar{s} \leq u^R_t, s \leq u^S_t \mbox{ and } \bar{s}+s \leq d_{tT}\\
+\infty, & \mbox{ otherwise.}
\end{aligned}
\right.
$$ 

\par Let $H_t(s^R_{t-1},s^S_{t-1})$ be the minimum cost of satisfying the demands $d_{tT}$ where $s^R_{t-1}$ (resp. $s^S_{t-1}$) units are stored at period $t-1$ at the retailer (resp. supplier) level. 
From the definition of the cost $V_t(s^R_{t-1},\bar{s},s^S_{t-1},s)$, we can compute the cost $H_t(s^R_{t-1},s^S_{t-1})$ as follows: 
$$
H_t(s^R_{t-1},s^S_{t-1})=\min_{\bar{s}\in \mathcal{S}^R_{t}, s \in \mathcal{S}^S_{t} } \{ V_t(s^R_{t-1},\bar{s},s^S_{t-1},s) + H_{t+1}(\bar{s},s) \},
$$
where $\mathcal{S}_{t}^R = \{ \max(0,s^R_{t-1}-d_t),\ldots, M^R_t\}$, with $M^R_t = \min(u^R_t,d_{tT})$, and $\mathcal{S}_{t}^S = \{ \max(0,s^S_{t-1}-X^R),\ldots, M^S_t\}$, with $M^S_t = \min(u^S_t,d_{tT})$. 
\\

\noindent \textbf{Optimal cost} 
\par The optimal cost of satisfying the demands $d_{1T}$ assuming that $s^R_0 = s^S_0 = 0$ is given by $H_{1}(0,0)$. 
We initialize the recursion by setting $H_{T+1}(s^R_{t},s^S_{t}) = 0$  for all the values $s^R_{t-1}$ and $s^S_{i-1}$ ensuring feasibility. \\

\noindent \textbf{Complexity analysis} 
\par Computing the cost $V_t(s^R_{t-1},\bar{s},s^S_{t-1},s)$ can be done in $\mathcal{O}(u^S_{t-1} u^R_{t-1} M^R_t M^S_t)$ for each period $t$. 
Therefore, it takes $\mathcal{O}(\sum_{i=1}^T (u^S_{i-1} u^R_{i-1} M^R_i M^S_i))$ to compute the optimal cost $H_1(0,0)$. 
This bound constitutes the complexity of the dynamic programming algorithm.
This is pseudo-polynomial, implying that the \dulsibsr problem is not strongly NP-hard. 
\\

\par In the next section, we consider the 2ULS problems with inventory bounds assuming that the demand at the retailer level has to be covered by a single order.

\section{Analysis of lot-sizing problems without lot-splitting} \label{sec:2uls_nls}
\par 
aruphongsa \etal \cite{Jaruphongsa2004} introduce the
problem where each demand must be satisfied by exactly one dispatch, \ie the demand lot-splitting is not allowed at the retailer level. 
We called this constraint the No Lot-Splitting (NLS) constraint. 
In practice, this study is motivated by traceability requirements for the product where the management of the inventory and the transport can be improved 
if the demand is supplied from the supplier to the retailer by a single delivery. 
We note $x^R_{kt}\geq 0$ the quantity of demand $d_t$ which is ordered at period $k$ to satisfy a demand $d_t$ at the retailer level. We have $\sum_{i=1}^t x^R_{it} = d_t$.

\begin{defn}[NLS constraint] \label{def:nls}
An ordering plan $x^R$ fulfills the NLS constraint if there does not exist two periods $l$ and $k$ with $l < k \leq t$ such that $x^R_{lt} > 0$ and $x^R_{kt} > 0$ for all periods $t$.
\end{defn}

\par The \dulsibr and the \dulsibs problems with the NLS constraint are denoted by \dulsibrNLS and \dulsibsNLS respectively. 
Before studying the complexity of the latter problems, it is interesting to analyze the complexity of the single level problem with NLS constraint,
that we denote by \ulsNLS.

\subsection{The \ulsNLS problem} \label{sec:uls_nls}
\par We consider $T$ periods $\{1,\dots,T\}$. In the  \ulsNLS problem, ordering units at period $t$ induces a fixed ordering cost $f_t$ and a unit ordering cost $p_t$. Carrying units from period $t$ to period $t+1$ induces a holding cost $h_t$. The total cost is given by the sum of the ordering and holding costs. The aim is to determine an ordering plan which satisfies the demands and which minimizes the total cost. We denote by $x_t$ the ordering quantity at period $t$, $s_t$ the inventory quantity at the end of period $t$ and $y_t$ the binary (setup) variable which is equal to 1 if there is an order at period $t$ and 0 otherwise. 
We say that the inventory bound is \textit{stationary} if $u_t$ is constant throughout the planning horizon. 

\begin{thm} \label{thm:uls_ib_nls}
The \ulsNLS problem is strongly NP-hard, even if the inventory bound is stationary.
\end{thm}

The detail of the proof can be found in the appendix.
The proof is based on a reduction from the 3-Partition problem which is strongly NP-hard~\cite{gareyjohnson}. 

\par Note that \ulsib problem can be solved in polynomial time \cite{AtamturkAlgo,Love1973}. Theorem \ref{thm:uls_ib_nls} shows that adding the NLS constraint to this problem makes it strongly NP-hard. 

\subsection{The \dulsibrNLS problem} \label{sec:duls_ibr_nls}

\begin{thm} \label{thm:duls_ibr_nls}
The \dulsibrNLS problem is strongly NP-hard, even if the inventory bound is stationary.
\end{thm}

\begin{proof}
We do a reduction from the \ulsNLS problem, that is strongly NP-hard, as shown by Theorem~\ref{thm:uls_ib_nls}. 
We transform an instance of the \ulsNLS problem into the following instance of the \dulsibrNLS problem. The costs of the retailer are the ones
of the \ulsNLS problem, \ie $u^R_t = u_t, f^R_t = f_t, p^R_t = p_t$ and $h^R_t = h_t$ for all $t \in \{1,\ldots,T\}$. The supplier costs 
are given by $f^S_t = h^S_t = p^S_t = 0$ for all $t \in \{1,\ldots,T\}$. The demands are the same than the ones of the \ulsNLS problem. 
Since all the supplier's costs are 0, the cost of an optimal solution for the \ulsNLS problem is equal to the optimal cost of its corresponding \dulsibrNLS instance (see Figure \ref{fig:dulsibr_nls_sol}). 
\begin{figure}[H]
    \centering
        \includegraphics[clip=true, scale = 0.85]{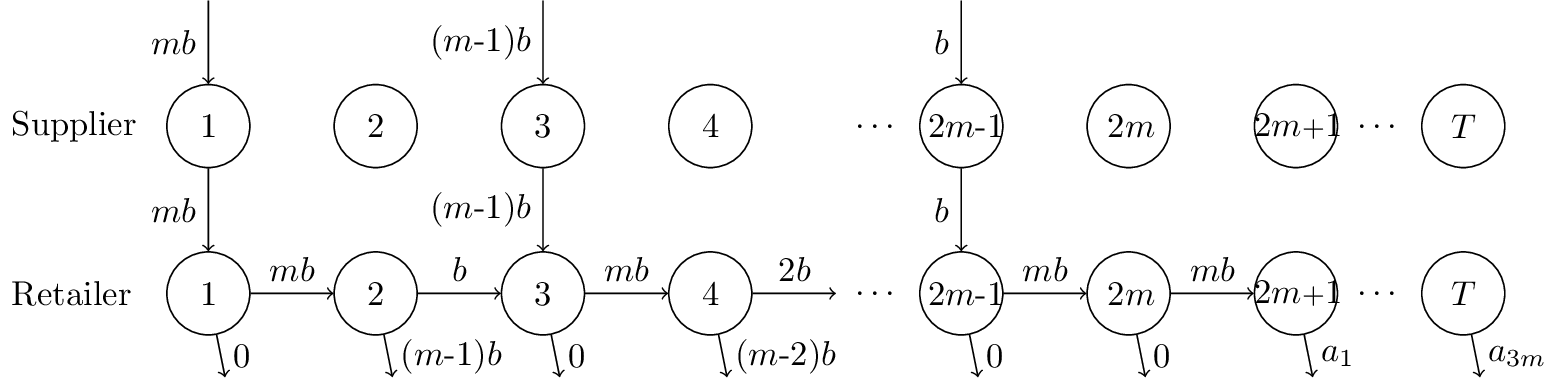}
        \caption{Solution for the \dulsibrNLS problem in the proof of Theorem \ref{thm:duls_ibr_nls}.}
        \label{fig:dulsibr_nls_sol}
\end{figure}

\par By Theorem \ref{thm:uls_ib_nls}, the \dulsibrNLS problem is also strongly NP-hard.
\end{proof}

\subsection{The \dulsibsNLS problem}  \label{sec:duls_ibs_nls}
\par Jaruphongsa \etal \cite{Jaruphongsa2004} prove that the \dulsibsNLS problem with demand time window constraints is weakly NP-hard. We show that this problem is also weakly NP-hard without demand time window constraints, and that it is even strongly NP-hard.

\begin{thm} \label{thm:duls_ibs_nls}
The 2ULS-IB$_{S}$-NLS is strongly NP-hard, even if the inventory bound is stationary.
\end{thm} 

\begin{proof}
As in the proof of Theorem~\ref{thm:duls_ibr_nls}, we do a reduction from the \ulsNLS problem, which is strongly NP-hard, as shown in Theorem~\ref{thm:uls_ib_nls}. We transform an instance of the \ulsNLS problem into the following instance of problem \dulsibsNLS. The supplier's costs are the ones 
of the \ulsNLS problem, \ie $u^S_t = u_t, f^S_t = f_t, p^S_t = p_t$ and $h^S_t = h_t$ for all $t \in \{1,\ldots,T\}$. The retailer's costs are given by $f^R_t = p^R_t = 0$ for all $t \in \{1,\ldots,T\}$ 
and, for all $t \in \{1,\ldots,T\}$,  $h^R_t = M$, where $M$ is a large number (we can fix $M=\sum_{t=1}^T (h_t + p_t)$). By this way, in an optimal solution of the \dulsibsNLS problem, no quantity will be stored at the retailer level.  The demands are the same than the ones of the \ulsNLS problem. Since $f_t^R=p_t^R=0$, the cost of an optimal solution of the \dulsibsNLS problem is equal to the optimal cost of its corresponding \ulsNLS problem. Figure~\ref{fig:dulsibs_nls_sol} illustrates such a solution.  
\begin{figure}[H]
    \centering
        \includegraphics[clip=true, scale = 0.85]{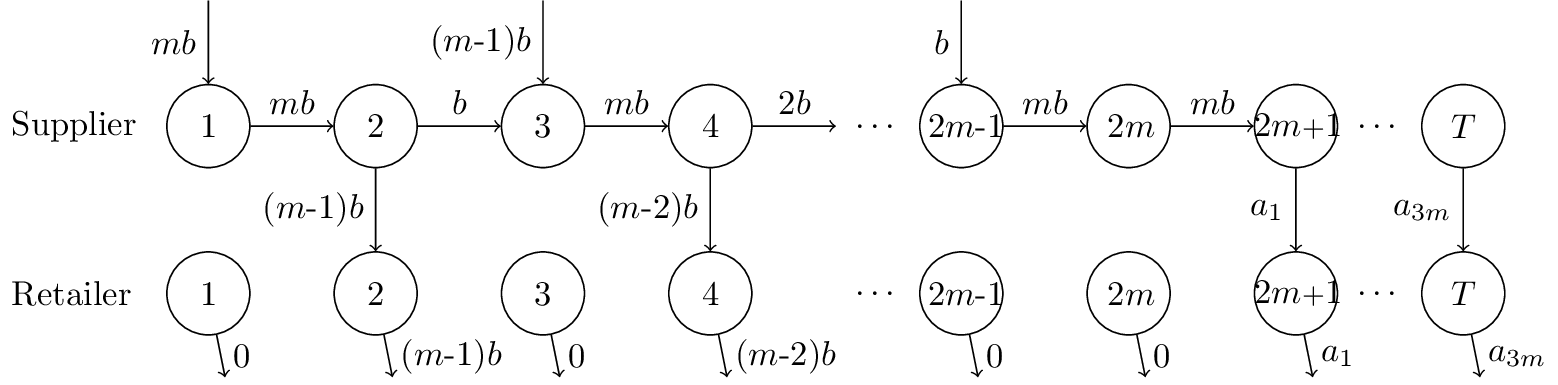}
        \caption{Solution for the \dulsibsNLS problem in the proof of Theorem \ref{thm:duls_ibs_nls}.}
        \label{fig:dulsibs_nls_sol}
\end{figure}
 
\par Therefore, by Theorem \ref{thm:uls_ib_nls}, the \dulsibsNLS problem is also strongly NP-hard.
\end{proof}

\subsection{The \dulsibsrNLS problem} \label{sec:duls_ibsr_nls}
\par Consider the case where the supplier and the retailer have inventory bounds. We prove that the 2ULS-IB$_{SR}$-NLS problem is strongly NP-hard.

\begin{thm}
The \dulsibsrNLS problem is strongly NP-hard, even if the inventory bound is stationary.
\end{thm}

\begin{proof}
The proof of this theorem is the same as the one of Theorem \ref{thm:duls_ibs_nls} for the \dulsibsNLS problem by adding any inventory bound at the retailer level (in an optimal solution no quantity will be stored at the retailer level).
\end{proof}

\section{Conclusion and future work}
\par This paper considers two-level uncapacitated lot-sizing problems with inventory bounds, 
and provides a complexity analysis of these problems. We present an $\mathcal{O}(T^4)$ dynamic programming algorithm which solves the problem where the inventory bounds are set at the retailer level. 
When the inventory bounds are set at the supplier level, we prove that the problem is weakly NP-hard. We also present a pseudo-polynomial dynamic programming algorithm which ensures that this problem is not strongly NP-hard. Considering that lot-splitting is not allowed, we prove that the ULS problem with inventory bounds and the 2ULS problems where the inventory bounds are set either at the retailer level, or at the supplier level or at both of them are strongly NP-hard. 

\par The following tables summarize the complexity results for 2ULS-IB problems:
\begin{table}[H]
\begin{center}
\begin{tabular}{|p{2.8cm}|p{7.5cm}|}
\hline
Problem & Complexity \\ \hline
ULS-IB & polynomial \cite{AtamturkAlgo}, \cite{Love1973} \\ \hline
\dulsibr & polynomial (Section \ref{sec:2uls_ibr}) \\ \hline
\multirow{2}*{\dulsibs} & polynomial with particular cost structure~\cite{Jaruphongsa2004} \\ \cline{2-2}
 & NP-hard(Section \ref{sec:2uls_ibs})\\ \hline
\dulsibsr & NP-hard (Section \ref{sec:2uls_ibsr})  \\ \hline
\end{tabular}
\end{center}
\caption{Complexity results with lot-splitting}
\end{table}
\begin{table}[H]
\begin{center}
\begin{tabular}{|p{3.6cm}|p{7.5cm}|}
\hline
Problem &  Complexity \\ \hline
ULS-NLS  & strongly NP-hard (Section \ref{sec:uls_nls}) \\ \hline
\dulsibrNLS & strongly NP-hard (Section \ref{sec:duls_ibr_nls}) \\ \hline 
\multirow{2}*{\dulsibsNLS}  & weakly NP-hard with demand time windows~\cite{Jaruphongsa2004}\\ \cline{2-2}
 &  strongly NP-hard (Section \ref{sec:duls_ibs_nls}) \\ \hline
2ULS-IB$_{SR}$-NLS  & strongly NP-hard (Section \ref{sec:duls_ibsr_nls})\\ \hline
\end{tabular}
\end{center}
\caption{Complexity results without lot-splitting}
\end{table}

\par It would be interesting for a future work to improve the running time of the algorithm solving the \dulsibr problem. 
Moreover, the complexity of the \dulsibs problem where the inventory bounds of the supplier 
are stationary is an open problem. Another interesting perspective is to consider that the supplier and the retailer share the same inventory facility. 
In this case, at each period, the inventory quantity of the supplier plus the one of the retailer cannot exceed a given inventory bound. 
The lot-sizing problems that have been studied is this paper consider a single item. It would also be interesting to study the case where there are several items. Finally, investigating efficient algorithms to solve the NP-hard 2ULS problems with inventory bounds is also a promising issue for practical applications.\\

\noindent \textbf{Acknowledgements} This work was supported by FUI project RCSM ``Risk, Credit Chain \& Supply Chain Management'', financed by R\'egion Ile-de-France.

\appendix
\makeatletter
\def\@seccntformat#1{Appendix~\csname the#1\endcsname}
\makeatother
\section{} \label{appendix:atamturk}
\noindent {\it Proof of Observation~\ref{obs1}. }
(i) If $\phi^{\alpha\beta}_{ijk} = +\infty$, then the \reg subplan $[i,j]$ with a single order at period $k$, $s^R_{i-1} = \alpha$ and $s^R_j = \beta$ is 
not a block. The violation(s) observed in the \reg subplan $[i,j]$ will also hold for the \reg subplan $[i-1,j]$. \\

\par (ii) If $\phi^{0\beta}_{ijk} \neq +\infty$, then $[i,j]^0_\beta$ is a block, and by Property~\ref{prop:prod_period}\ref{prop:prod_period_i} there is an ordering period at $k=i$. We consider the \reg subplan $[i-1,j]$ with an order at period $k$, $s^R_{i-2} = 0$ and $s^R_j=\beta$. If $d_{i-1} = 0$, then the \reg subplan $[i-1,j]$ is not a block since $s^R_{i-2} = 0$. If $d_{i-1} > 0$, then $d_{i-1}$ could not be covered and thus the \reg subplan $[i-1,j]$ is not a block. \\

\par (iii) If $\phi^{u_{i-1}\beta}_{ijk} \neq +\infty$, then  $[i,j]^{u_{i-1}}_\beta$ is a block with an ordering period $k$ if it exists. We consider the \reg subplan $[i-1,j]$ with a single order at period $k$, $s^R_{i-2} = u_{i-2}$ and $s^R_j = \beta$. We want to determine if this \reg subplan is a block. 
\par We know that $u_{i-2} \leq u_{i-1} + d_{i-1}$ (Assumption~\ref{obs:ib_prop}). If $u_{i-2} = u_{i-1} + d_{i-1}$, then this \reg subplan is not a block because in that case $s^R_{i-1} = u_{i-1}$.
If $u_{i-2} < u_{i-1} + d_{i-1}$, the we have $s^R_{i-1} < u_{i-1} \leq d_{ij} + \beta$, and there must be an ordering period at $k$ in the subplan $[i-1,j]$. 
\par Moreover, if $u_{i-2} > d_{i-1,k-1}$, then we have a block $[i-1,j]^{u_{i-2}}_\beta$. 
The retailer has to order a quantity $u_{i-1} - u_{i-2} + d_{i-1} > 0 $ 
in addition to $X_{ij} = d_{ij} - u_{i-1} + \beta$ at period $k$. The inventory quantities between periods $k$ and $j$ remain unchanged in the block $[i-1,j]^{u_{i-2}}_\beta$. 
Since the demand $d_{i-1}$ has to be covered by $u_{i-2}$, there are $u_{i-1} - u_{i-2} + d_{i-1}$ less units in the inventory between periods $i-1$ and $k-1$. 
The cost $\phi^{u_{i-2}\beta}_{i-1,j,k}$ of the block $[i-1,j]^{u_{i-2}}_\beta$ can be derived from $\phi^{u_{i-1}\beta}_{ijk}$ by considering these two cases:\\
Case 1: Assume that a quantity $X_{ij} > 0 $ is ordered at period $k$ in the block $[i,j]^{u_{i-1}}_\beta$. Then, the cost of the block $[i-1,j]^{u_{i-2}}_\beta$ is given by:
$\phi^{u_{i-1}\beta}_{ijk} + h^R_{i-2}u_{i-2} + (p^R_k - \sum_{l=i-1}^{k-1}h^R_l)(u_{i-1} - u_{i-2} + d_{i-1}) = \phi^{u_{i-1}\beta}_{ijk} + \Delta_1 $.\\
Case 2: Assume that no ordering period occurs in the block $[i,j]^{u_{i-1}}_\beta$. Then, an additional fixed ordering cost $f^R_k$ must be considered to compute 
the cost of the block $[i-1,j]^{u_{i-2}}_\beta$, which will be given by:
$\phi^{u_{i-1}\beta}_{ijk} + f^R_k + \Delta_1 = \phi^{u_{i-1}\beta}_{ijk} + \Delta_2$. 
\qed

\section{}

\noindent {\it Proof of Theorem~\ref{thm:uls_ib_nls}. }
We show that the 3-Partition problem, which is strongly NP-hard~\cite{gareyjohnson}, can be reduced to the \ulsNLS problem in polynomial time. Recall that an instance of the 3-Partition problem is given by an integer $b$ and $3m$ integers $(a_1, \ldots, a_{3m})$ such that $\sum_{i=1}^{3m} a_i = mb$ and $b/4 < a_i < b/2$ for all $i \in \{1,\ldots,3m\}$. The question is: does there exist a partition $A_1 \cup \ldots \cup A_m$ of $\{1, \ldots, 3m\}$ such that $\sum_{i \in A_j} a_i = b$ for all $j \in \{1,\ldots, m\}$. 

\par We transform an instance of the 3-Partition problem into an instance of the \ulsNLS problem in the following way:
\begin{itemize}
\item[-] $T=5m$ periods. Let us note $\mathcal{T}_1$ (resp. $\mathcal{T}_2$) the set of odd (resp. even) periods in the set $\{1,\ldots,2m\}$. 
\item[-] $d_t = 0$ for all $t \in \mathcal{T}_1$ \\ $d_t = (m-t/2)b$ for all $t \in \mathcal{T}_2$ \\ $d_t = a_{t-2m}$ for all $t \in \{2m+1,\ldots,T\}$
\item[-] $f_t = 0$ for all $t \in \mathcal{T}_1$ \\ $f_t = 1$ for all $t \in \mathcal{T}_2 \cup \{2m+1,\ldots,T\}$
\item[-] $h_t = p_t = 0$ for all $t \in \{1,\ldots,T\}$
\item[-] $u_t = mb$ for all $t \in \{1,\ldots,T\}$
\end{itemize}
The instance is illustrated in Figure \ref{fig:uls_nls_inst}. The fixed ordering costs are indicated at the top of each period. The inventory bounds are represented on the horizontal edges.
\begin{figure}[H]
    \centering
        \includegraphics[clip=true, scale = 0.85]{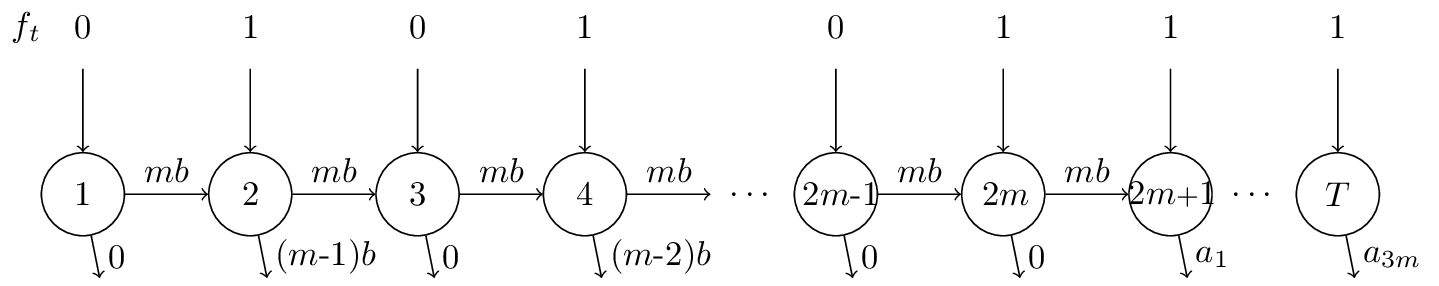}
        \caption{Instance of the \ulsNLS problem in the proof of Theorem \ref{thm:uls_ib_nls}.}
        \label{fig:uls_nls_inst}
\end{figure}

\par Let us show that there exists a solution to the \ulsNLS problem of cost at most $0$ if and only if there exists a solution to the 3-Partition problem.\\

\par Assume that there exists a solution $(A_1,\ldots,A_{3m})$ of the 3-Partition problem. 
The cost of the following solution of the \ulsNLS problem is $0$: at each period $t \in \mathcal{T}_1$, we order $x_t = \sum_{i \in A_{(t+1)/2}} a_i + d_{t+1} = b + b \left( m - \frac{t+1}{2} \right)$ units. Since the ordering cost is equal to 0 for all $t \in \mathcal{T}_1$, it costs 0 to order these units. At each period $t \in \mathcal{T}_2$, the demand $d_t$ is satisfied and $b$ units are stored which implies that there is exactly $s_t = \frac{t}{2}b$ units in stock at the end of period $t$. At each period $t \in \mathcal{T}_1$, we store exactly a quantity $s_{t-1} + x_t = \frac{t-1}{2}b + \left( m - \frac{t-1}{2}\right)b = mb$ and the inventory bound $u_t$ is not exceeded. Each demand $d_t$ for all $t<2m$ is satisfied and there is $mb$ units in stock at period $2m$ for satisfying the demands at period $\{2m+1,\ldots,T\}$. Since there is no holding cost, the cost of this solution is 0. Note that this solution fulfills the NLS constraint since each demand is satisfied by a single order.

\par Assume now that there exists a solution to the \ulsNLS problem of cost at most 0 (see Figure \ref{fig:uls_nls_sol}). 
\begin{figure}[H]
    \centering
        \includegraphics[clip=true, scale = 0.85]{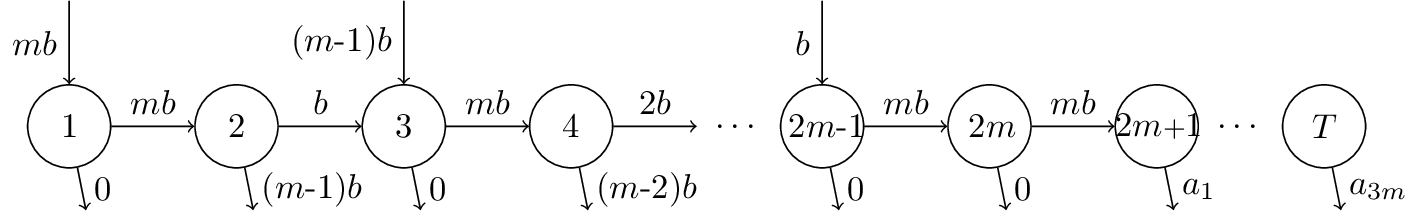}
        \caption{Solution for the \ulsNLS problem in the proof of Theorem \ref{thm:uls_ib_nls}.}
        \label{fig:uls_nls_sol}
\end{figure}
Since the fixed ordering cost is equal to $1$ for all $t \in \mathcal{T}_2 \cup \{2m+1,\ldots,T\}$, we cannot order at these periods. Thus, all orders are set at period $t \in \mathcal{T}_1$. Since for each period $t \in \mathcal{T}_2$, $d_t = (m-t/2)b$, and since the inventory bound is $mb$, at most $\frac{t}{2}b$ units can be stored from period $t \in \mathcal{T}_2$ to a period in $\mathcal{T}_1$ . Since $2m$ units have to be available at period $2m$ (otherwise the cost will be greater than 0), then $\frac{t}{2}b$ units have to be stored from period $t \in \mathcal{T}_2$ to period $t+1$. So, we have to order $b$ units at each period $t \in \mathcal{T}_1$ for satisfying the demands $d_{2m+1,T}$.  Assuming the NLS constraint, each demand $d_t$ for all $t \in \{2m+1, \ldots,T\}$ is satisfied by a single ordering period at $t \in \mathcal{T}_1$. So, there is a partition of the periods $\{2m+1, \ldots,T\}$ into $m$ sets $(A_1,\ldots,A_m)$ such that $\sum_{i \in A_j} d_i = b $ for all $j \in \{1,\ldots,m\}$. Since each demand $d_t$ for all $t \in \{2m+1, \ldots,T\}$ corresponds to an integer of $(a_1, \ldots,a_{3m})$, this means that there exists a solution to the 3-Partition problem.
\qed

\bibliographystyle{elsarticle-num}
\bibliography{biblio}

\end{document}